\documentclass{article}
\pdfoutput=1
\usepackage[T1]{fontenc}
\usepackage[utf8]{inputenc}
\usepackage{microtype}
\usepackage[total={6in, 8in}]{geometry}
\usepackage{comment}
\usepackage{amsmath}
\usepackage{amssymb}
\usepackage{amsthm}
\usepackage{graphicx}
\usepackage{mathrsfs}
\usepackage{braket}
\usepackage{bbm}
\usepackage{xcolor}
\usepackage{cite}
\usepackage{doi}
\usepackage{tikz}
\usepackage{booktabs}
\usepackage{threeparttable}
\usepackage{multirow}

\newtheorem{theorem}{Theorem}[section]
\newtheorem{lemma}[theorem]{Lemma}
\usetikzlibrary{quantikz2}
\usepackage{theoremref}

\usepackage{algorithm}
\usepackage{algorithmic}

\newtheorem{thm}{\protect\theoremname}
\theoremstyle{plain}
\newtheorem{lem}{\protect\lemmaname}
\theoremstyle{plain}

\theoremstyle{plain}
\newtheorem*{lem*}{\protect\lemmaname}
\theoremstyle{plain}
\newtheorem*{thm*}{\protect\theoremname}
\theoremstyle{plain}

\theoremstyle{plain}

\theoremstyle{plain}
\newtheorem*{cor*}{\protect\corollaryname}

\newtheorem*{defn*}{Definition}
\newtheorem{prob}{Problem}

\newtheorem{fact}{Fact}

\usepackage[USenglish]{babel}
  \providecommand{\corollaryname}{Corollary}
  \providecommand{\lemmaname}{Lemma}
  \providecommand{\propositionname}{Proposition}
  \providecommand{\remarkname}{Remark}
\providecommand{\theoremname}{Theorem}

\usepackage[normalem]{ulem}
\usepackage{authblk}

\hypersetup{colorlinks=true, linkcolor=blue, urlcolor=blue, citecolor = blue}

\newcommand{\Tr}{\mathrm{Tr}}

\title{Improved Hamiltonian learning and sparsity testing through Bell sampling}
\author[1,2]{Savar D. Sinha}
\affil[1]{Department of Computing and Mathematical Sciences, California Institute of Technology}
\affil[2]{The Division of Physics, Mathematics and Astronomy, California Institute of Technology}
 \author[3,4,5]{Yu Tong}
    \affil[3]{Department of Mathematics, Duke University}
    \affil[4]{Department of Electrical and Computer Engineering, Duke University}
    \affil[5]{Duke Quantum Center, Duke University}

\begin{document}

\maketitle

\begin{abstract}
   We consider the problem of learning an $M$-sparse Hamiltonian and the related problem of Hamiltonian sparsity testing. Through a detailed analysis of Bell sampling, we reduce the total evolution time required by the state-of-the-art algorithm for $M$-sparse Hamiltonian learning to $\widetilde{\mathcal{O}}(M/\epsilon)$, where $\epsilon$ denotes the $\ell^{\infty}$ error, achieving an improvement by a factor of $M$ (ignoring the logarithmic factor) while only requiring access to forward time-evolution. We then establish a connection between Hamiltonian learning and Hamiltonian sparsity testing through Bell sampling, which enables us to propose a Hamiltonian sparsity testing with state-of-the-art total evolution time scaling. 
\end{abstract}



\section{Introduction}
\label{sec:intro}


Determining the Hamiltonian that governs the dynamics is essential for understanding and controlling quantum systems. 
In quantum simulation and quantum error correction, knowledge of the Hamiltonian permits better optimization and control. 
This task is closely related to classical tasks, such as learning graphical models or Boolean functions, providing rich opportunities to explore how the quantum world connects to and contrasts with the classical one.
We will hereafter refer to the task of identifying the Hamiltonian from dynamics as \textit{Hamiltonian learning.} A related task, learning the Hamiltonian from the equilibrium state of the quantum system, is not considered in this work.

Many fundamental features of quantum mechanics make Hamiltonian learning challenging. The non-commutativity among the Hamiltonian terms makes it hard to isolate their individual effects. Quantum dynamics creates entanglement, making it difficult for classical computers to predict the outcomes of experiments.  
Given the importance of this task and the many challenges associated with it, Hamiltonian learning has been the focus of many previous works \cite{Seif2021compressed,evans2019scalablebayesianhamiltonianlearning,li2020hamiltonian,che2021learning,HaahKothariTang2022optimal,yu2023robust,hangleiter2024robustlylearninghamiltoniandynamics,StilckFrança2024,ZubidaYitzhakiEtAl2021optimal,BaireyAradEtAl2019learning, bairey2020learning,GranadeFerrieWiebeCory2012robust,gu2022practical,wilde2022learnH,KrastanovZhouEtAl2019stochastic,Caro_2024,MobusBluhmCaroEtAl2023dissipation,HolzapfelEtAl2015scalable, HuangTongFangSu2023learning,dutkiewicz2023advantage,MiraniHayden2024learning,NiLiYing2024quantum,LiTongNiGefenYing2023heisenberg,BoixoSomma2008parameter,bakshi2024structure,WangLi2024simulation,odake2023universal}. 
Most early studies on Hamiltonian learning aimed to achieve Heisenberg-limited scaling, i.e., estimating Hamiltonian parameters to precision $\epsilon$ with a total cost of $\mathcal{O}(1/\epsilon)$ \cite{higgins2007entanglement}. These works were either restricted to small quantum systems or addressed scalability in large systems but without reaching the Heisenberg limit. 
A recent advance by \cite{HuangTongFangSu2023learning} introduced an approach for learning an $n$-qubit many-body Hamiltonian with total evolution time $\mathcal{O}(\log(n)/\epsilon)$, where $\epsilon$ denotes the $\ell^\infty$-error in the coefficients. This result holds under the assumption that the Hamiltonian is \emph{low-intersection}, meaning that each qubit participates in only a constant number of Pauli terms independent of $n$, and that the set of terms is known in advance. This condition can be regarded as a mild relaxation of geometric locality. The methods developed in \cite{HuangTongFangSu2023learning} have been extended to alternative control models \cite{dutkiewicz2023advantage} as well as to bosonic and fermionic systems \cite{LiTongNiGefenYing2023heisenberg,NiLiYing2024quantum,MiraniHayden2024learning}.

Later works expanded the class of Hamiltonians that can be efficiently learned to include $k$-local Hamiltonians that are $M$-sparse \cite{bakshi2024structure,ma2024learningkbodyhamiltonianscompressed,Zhao2025learning,hu2025ansatzfreehamiltonianlearningheisenberglimited,Arunachalam2025testing}. Here $k$-local means that each Pauli term in the Hamiltonian acts non-trivially on at most $k$ qubits, and $M$-sparse means there are at most $M$ Pauli terms with nonzero coefficients in the Hamiltonian. In all of the above works it is required that $M=\mathrm{poly}(n)$. In particular, compressed sensing plays a crucial role in \cite{ma2024learningkbodyhamiltonianscompressed}, and has been used in previous works on learning Lindbladian noise \cite{Seif2021compressed}.

Among all previous algorithms, the only other ones that do not assume the Hamiltonian to be $k$-local and still achieve the Heisenberg-limited scaling are Theorem 1.5 in Ref.~\cite{Zhao2025learning} and Ref.~\cite{hu2025ansatzfreehamiltonianlearningheisenberglimited}. 
The former proposes an algorithm with total evolution time $\widetilde{\mathcal{O}}(M/\epsilon)$ to learn all coefficients to within additive error $\epsilon$,\footnote{We use $\widetilde{O}$ to hide the $\mathrm{polylog}(n,M,1/\epsilon)$ factor.} but requires access to the time-reversal dynamics of the Hamiltonian,\footnote{Time-reversal dynamics is sometimes used to refer to $e^{-iH^\top t}$, but here we use it in a way consistent with \cite{Zhao2025learning}.} i.e., in addition to the forward dynamics $e^{-iHt}$, they also require access to $e^{iHt}$, which is not usually available in an experimental setting (although it is possible to implement it from $e^{-iHt}$ with significant overhead \cite{Odake2024higher}). 
The latter algorithm uses only forward dynamics $e^{-iHt}$, but comes with a total evolution time scaling of $\widetilde{\mathcal{O}}(M^2/\epsilon)$.

In this work, through a more detailed analysis of the Bell sampling procedure used extensively in \cite{hu2025ansatzfreehamiltonianlearningheisenberglimited}, we improve the total evolution time estimate of their algorithm to $\widetilde{\mathcal{O}}(M/\epsilon)$ as stated in Theorem~\ref{thm:intol_learn}. This matches the total evolution time needed in \cite{Zhao2025learning}, with the advantage that the resulting Hamiltonian learning algorithm does not employ experimentally difficult time-reversal dynamics.


In classical learning theory, learning is closely related to the problem of property testing \cite{GoldreichGoldwasserRon1998property}. For Hamiltonian learning, one can naturally consider the corresponding Hamiltonian property testing problems, such as testing whether the Hamiltonian is $k$-local, which is known as \textit{locality testing}, or whether the Hamiltonian can be expressed as a linear combination of $M$ Pauli operators \cite{bluhm2024hamiltonian,Arunachalam2025testing,KallaugherLiang2025hamiltonian}, which is known as \textit{sparsity testing}. These testing problems are of practical interest, as these properties determine which Hamiltonian learning algorithm can best be used to learn the target Hamiltonian. Moreover, the sparsity assumption is necessary for all Heisenberg-limited Hamiltonian learning algorithms that we know today.

For Boolean functions, it has been shown that a \textit{proper} learning algorithm implies a testing algorithm with a similar runtime \cite{GoldreichGoldwasserRon1998property}. For the setting of Hamiltonian learning, such a connection has not been rigorously established in the literature to the best of our knowledge. In this work, we establish such a connection through Bell sampling for Hamiltonian sparsity testing. For an unknown Hamiltonian $H$, we apply the Hamiltonian learning algorithm in \cite{hu2025ansatzfreehamiltonianlearningheisenberglimited}, whose analysis has been improved in this work, to learn an approximation $\hat{H}$. We then engineer the time evolution under the difference $H-\hat{H}$ through Trotterization, which enables us to estimate the normalized Frobenius norm $\|H-\hat{H}\|_F$ through Bell sampling. For $M$-sparse Hamiltonians, $\|H-\hat{H}\|_F$ is guaranteed to be small, while if a Hamiltonian is far from $M$-sparse, then $\|H-\hat{H}\|_F$ is bounded away from 0.

This approach enables us to obtain an algorithm with the best total evolution time scaling among all the algorithms we know. As stated in Theorem~\ref{thm:intol_sparse}, for a Hamiltonian $H$ satisfying $\|H\|\leq 1$, sparsity testing can be done correctly with high probability with total evolution time $\widetilde{\mathcal{O}}(M^{1.5}/\epsilon+M/\epsilon^2)$. Compared to the best previous result in \cite{Arunachalam2025testing} (which focuses on the harder problem of tolerant testing but has a result for the non-tolerant setting in Theorem 4.4 in the arXiv version \cite{ArunachalamDuttEscudero2024arXiv}), which has a total evolution time $\mathcal{O}(M^{1.5}/\epsilon^3)$, we achieve a speedup with a factor of $\min\{1/\epsilon^2,\sqrt{M}/\epsilon\}$. Our technique for connecting the Hamiltonian property testing problem to a Hamiltonian learning algorithm extends beyond sparsity testing, and we expect it to be useful for many other Hamiltonian testing problems.


    
\paragraph{Notations.}

In this work, we denote the identity operator $I$ and the Pauli operators $X, Y, Z$ in terms of their 2-bit representations $\sigma_{00}, \sigma_{10}, \sigma_{11}, \sigma_{01}$, respectively. In general, for $a, b \in \{0, 1\}^n$, we define
$$\sigma_{ab} = \bigotimes_{j = 1} ^n \sigma_{a_jb_j} = \bigotimes_{j = 1} ^n i^{a_jb_j}X^{a_j}Z^{b_j}$$

We describe sampling procedures involving 2-qubit Bell-states, which are defined as follows:
\begin{align*}
    |\Phi^+\rangle &= \frac{|00\rangle + |11\rangle}{\sqrt 2} & |\Phi^-\rangle &= \frac{|00\rangle - |11\rangle}{\sqrt 2} \\
    |\Psi^+\rangle &= \frac{|01\rangle + |10\rangle}{\sqrt 2} & |\Psi^-\rangle &= \frac{|01\rangle - |10\rangle}{\sqrt 2}
\end{align*}
As a shorthand, we will also use $|\Phi\rangle := |\Phi^+\rangle^{\otimes n}$. We can write any traceless Hamiltonian $H$ in terms of its Pauli decomposition as follows:
$$H = \sum_x \mu_x P_x$$
where $\mu_x \in \mathbb R, P_x \in \mathbb P_n \setminus \{I^{\otimes n}\}$. Here, $\mathbb P_n$ denotes the $n$-qubit Pauli group.

We also utilize several different norms to characterize various aspects of any Hamiltonian $H$. We define the normalized Frobenius norm $\lVert \cdot \rVert_F$ as follows:
$$\lVert H \rVert_F = \sqrt\frac{\Tr[H^\dagger H]}{2^n} = \sqrt{\sum_x |\mu_x|^2}$$
Note that the unnormlized Frobenius norm of $H$ yields the $\ell^2$-norm of the eigenvalue spectrum while the normalized Frobenius norm gives the $\ell^2$-norm of the Pauli coefficients. We also utilize the spectral norm $\lVert \cdot \rVert_2$, defined for $H = \sum_{i} \lambda_i |\xi_i\rangle \langle \xi_i|$ as
$$\lVert H \rVert_2 = \max_{i} |\lambda_i|$$

\section{Bell Sampling Bounds}

For any given unitary $U = \sum_x U_x \sigma_x$, we can use Bell sampling to sample a Pauli $\sigma_x$ with probability $|U_x|^2$. This is formalized and proven below as follows.
\begin{fact}\label{fact:bell_sampling}
    Given a unitary $U = \sum_x U_x\sigma_x$, we can sample a Pauli $\sigma_x$ with probability $|U_x|^2$.
\end{fact}

\begin{proof}
    Suppose we are given a $n$ qubit unitary $U = \sum_x U_x\sigma_x$. Using a system of $n$ qubits as well as $n$ ancillas, we can sample from the Pauli distribution of $U$ as follows. Initialize a state consisting of $n$ EPR pairs between the system and the ancillas: 
    $$|\psi\rangle = |\Phi^+\rangle^{\otimes n}$$
    If we now apply $U \otimes I_{2^n}$, we get the following:
    \begin{align*}
        (U \otimes I_{2^n})|\psi\rangle &= U \otimes I_{2^n}|\Phi^+\rangle^{\otimes n} \\
        &= \sum_{x \in \{00, 01, 10, 11\}^n} U_x \left(\bigotimes_{i \in [n]} \sigma_{x_i} \otimes I_{2^n}\right)|\Phi^+\rangle^{\otimes n}
    \end{align*}
    Here, $\sigma_{ab} = i^{ab}X^aZ^b$. We then note that $(\sigma_{01}\otimes I)|\Phi^+\rangle = |\Phi^-\rangle, (\sigma_{10}\otimes I)|\Phi^+\rangle = |\Psi^+\rangle, (\sigma_{11}\otimes I)|\Phi^+\rangle = i|\Psi^-\rangle$. Consequently, every one of the $4^n$ Pauli operators corresponds to a unique choice of $n$ Bell states. Since these states are all orthogonal, measuring in the Bell basis is akin to sampling a Pauli $\sigma_x$ with probability $|U_x|^2$.
\end{proof}

Naturally, this lends itself to the following fact relating the norm of an operator acting on a Bell state to the normalized Frobenius norm of said operator.
\begin{fact}\label{fact:bell_norm}
    For any $n$-qubit operator $A$ that admits a Pauli decomposition, we have that
    $$\lVert (A \otimes I_{2^n})|\Phi\rangle \rVert = \lVert A \rVert_F$$
\end{fact}

\begin{proof}
    Writing $A$ in terms of its Pauli decomposition $(A = \sum_x \mu_x\sigma_x)$, we have from the proof for Fact \ref{fact:bell_sampling} that
    $$A \otimes I_{2^n}|\Phi\rangle = \sum_{x \in \{00, 01, 10, 11\}^n} \mu_x \left(\bigotimes_{i \in [n]} \sigma_{x_i} \otimes I_{2^n}\right)|\Phi\rangle$$
    Since the set of $\left(\bigotimes_{i \in [n]} \sigma_{x_i} \otimes I_{2^n}\right)|\Phi\rangle$ taken over all $x$ forms an orthonormal basis, we have that
    $$\lVert (A \otimes I_{2^n})|\Phi\rangle \rVert = \sqrt{\sum_x |\mu_x|^2} = \lVert A \rVert_F$$
\end{proof}

Using Bell-sampling, we can sample from the Pauli distribution of the time-evolution operator $U(t) = e^{-iHt}$, which for small $t$, is simply approximated as
$$U(t) = I - iHt + O(t^2)$$
Since the Hamiltonian is traceless, this means that any non-identity term sampled is a part of the Hamiltonian with high probability. Consequently, since sampling the identity element gives us no information about the Hamiltonian itself, we need to bound the probability of sampling the identity element $I$ for sufficiently small $t$. These upper and lower bounds are formalized in the following theorem.

\begin{thm}\label{thm:identity_bounds}
    For a traceless Hamiltonian $H$ with bounded spectral norm $\lVert H \rVert_2 \le L$ for $L > 0$, the probability of sampling $I$, $\Pr[I]$ after evolving a system according to $U(t) = e^{-iHt}$ can be bounded as follows:
    $$0 \le 1 - \lVert H \rVert_F^2t^2 \le \Pr[I] \le 1 - 2c\lVert H \rVert_F^2t^2$$
    for any $c \in (0, 1/2)$ and $t \le \frac{t^*(c)}{2L}$, where $t^*(c) \in (0, 2\pi)$ satisfies $\cos(t^*(c)) = 1 - c(t^*(c))^2$.
\end{thm}

\begin{proof}
    Suppose we have the time-evolution operator $U(t) = e^{-iHt} = \sum_x U_x\sigma_x$. According to Fact \ref{fact:bell_sampling}, we have that
    $$\Pr[I] = |U_I|^2$$
    To compute $U_x$, we can take the trace of $U(t)$, as the traceless terms will vanish, giving us that
    $$\Tr[U(t)] = \sum_x U_x\Tr[\sigma_x] = 2^nU_I \implies U_I = \frac{\Tr[U(t)]}{2^n}$$
    Therefore, we have that the probability of sampling the identity is given as
    $$\Pr[I] = |U_I|^2 = \frac{|\Tr[U(t)]|^2}{4^n}$$

    To evaluate the trace of $U(t)$, we first write out the spectral decomposition of $H$:
    $$H = \sum_j \lambda_j |\xi_j\rangle \langle \xi_j|,$$
    where $\sum_j \lambda_j = 0$ (traceless), $\max |\lambda_j| < 1$ (bounded spectral norm), and $\sum_j \lambda_j^2 = 2^n\lVert H \rVert_F^2$. Note that each $\lambda_j \in \mathbb R$ since $H$ is Hermitian. By Sylvester's Formula, we have that
    $$U(t) = \sum_j e^{-i\lambda_j t} |\xi_j\rangle \langle \xi_j|$$
    Computing the trace of this expression, we have that
    $$\Tr[U(t)] = \sum_j e^{-i\lambda_j t}$$

    Substituting this value into the expression for the probability, we have the following:
    $$\Pr[I] = \frac{|\Tr[U(t)]|^2}{4^n} = \frac{(\sum_j e^{-i\lambda_j t})(\sum_j e^{i\lambda_j t})}{4^n} = \frac{\sum_{j, k} e^{-i(\lambda_j - \lambda_k) t}}{4^n}$$
    Since the numerator is $|\Tr[U(t)]|^2 \in \mathbb R$, we know that the imaginary part of the numerator vanishes, giving us that
    $$\Pr[I] = \frac{\sum_{j, k} \cos[(\lambda_j - \lambda_k) t]}{4^n}$$

    To compute the upper bound on this probability, we can observe that
    $$\cos[(\lambda_j - \lambda_k) t] \le 1 - c[(\lambda_j - \lambda_k) t]^2$$
    for $0 < c < 1/2$ and suitably small $|\lambda_j - \lambda_k| t < t^*(c)$, where $t^*(c) \in (0, 2\pi)$ is some constant only dependent on $c$ that satisfies
    $$\cos(t^*(c)) = 1 - c(t^*(c))^2$$
    By the spectral norm bound, we have that $|\lambda_j - \lambda_k| \le 2L$, meaning that any $t < \frac{t^*(c)}{2L}$, this bound holds. Expanding out this bounds gives us
    $$\Pr[I] = \frac{\sum_{j, k} \cos[(\lambda_j - \lambda_k) t]}{4^n} \le \frac{\sum_{j, k} 1 - c(\lambda_j - \lambda_k)^2t^2}{4^n} = 1 - \frac{ct^2}{4^n}\sum_{j, k}(\lambda_j - \lambda_k)^2$$
    Expanding out the quadratic expression, we have the following:
    $$\sum_{j, k}(\lambda_j - \lambda_k)^2 = \sum_{j, k}[\lambda_j^2 + \lambda_k^2 - 2\lambda_j\lambda_k] = 2^n\sum_{j}\lambda_j^2 + 2^n\sum_{k}\lambda_k^2 - 2\sum_{j}\lambda_j\sum_{k}\lambda_k = 2(4^n)\lVert H \rVert_F^2$$
    Substituting this back into the bound obtained previously gives
    $$\Pr[I] \le 1 - \frac{ct^2}{4^n}\sum_{j, k}(\lambda_j - \lambda_k)^2 = 1 - 2ct^2\lVert H \rVert_F^2$$

    Using an equivalent derivation, we can lower bound $\Pr[I]$ by using the following lower bound, which holds for all $\lambda_j, \lambda_k, t$:
    $$\cos[(\lambda_j - \lambda_k) t] \ge 1 - \frac{1}{2}[(\lambda_j - \lambda_k) t]^2$$
    Repeating the  same derivation gives us that
    $$\Pr[I] \ge 1 - t^2 \lVert H \rVert_F^2$$
\end{proof}

Using the upper bound on $\Pr[I]$ we can develop an efficient protocol for emptiness testing, while using the lower bound, we can demonstrate an improved sampling complexity for the sparsity learning algorithm presented in \cite{hu2025ansatzfreehamiltonianlearningheisenberglimited}.

\section{Emptiness Testing}

Suppose we are given a traceless Hamiltonian $H$ with bounded spectral norm $\lVert H \rVert_2 \le L$ for $L < 0$. A simple decision problem we can consider is whether such a Hamiltonian is empty or not. In general, we say that $H$ is empty if $\lVert H \rVert_F < \epsilon_1$ and it is not empty if $\lVert H \rVert_F > \epsilon_2$ for $\epsilon_2 > \epsilon_1$. 

\subsection{Intolerant Testing}

We first consider the intolerant testing case, where $\epsilon_1 = 0$ and $\epsilon_2 = \epsilon$. Formally, we have the following decision problem.

\begin{prob}[Intolerant Emptiness Testing]
    Given a traceless Hamiltonian $H$ with bounded spectral norm $\lVert H \rVert_2 \le L$, $L > 0$, that satisfies either $H = 0$ or $\lVert H \rVert_F \ge \epsilon$, decide if $H$ is either zero or $\epsilon$-far from zero.
\end{prob}

In order to distinguish between these two cases, we observe that Bell sampling will only return non-identity samples if the Hamiltonian is not empty. This approach is formalized in the following theorem and subsequent proof.

\begin{thm}\label{thm:intol_empt}
    There exists an algorithm that decides the intolerant emptiness testing for $H$ with probability of success $\ge 1 - \delta$ for a given threshold $\epsilon$ that requires $N$ queries and total time evolution $T$, where
    $$N = \mathcal O\left(\frac{L^2\log(1/\delta)}{\epsilon^2}\right)$$
    and
    $$T = \mathcal O\left(\frac{L\log(1/\delta)}{\epsilon^2}\right)$$
\end{thm}

\begin{proof}
    If $H = 0$, the time-evolution operator $U(t) = I$, meaning that Bell-sampling will always give the identity as an output.

    On the other hand, if $\lVert H \rVert_F > \epsilon$, we have that $H \ne 0$, meaning that there is a nonzero probability of getting non-identity samples. According to the upper bound on $\Pr[I]$ derived in Theorem~\ref{thm:identity_bounds}, we have that
    $$\Pr[I] \le 1 - 2ct^2\lVert H \rVert_F \le 1 - 2c\epsilon^2t^2$$
    Consequently, we have that the probability of sampling a non-identity element for $t = \mathcal O \left(\frac{1}{L}\right)$ is given as
    $$\Pr[\lnot I] = 2c\epsilon^2t^2 = \mathcal O\left(\frac{\epsilon^2}{L^2}\right)$$

    Using the Chernoff bound, we then have that to sample a non-identity element with probability $\ge 1 - \delta$, we need to use $N$ samples, where
    $$N = \mathcal O\left(\frac{L^2\log(1/\delta)}{\epsilon^2}\right)$$
    Since our evolution time is $\mathcal O(1/L)$ for each sample, we have that the evolution time is given as
    $$T = \mathcal O\left(\frac{L\log(1/\delta)}{\epsilon^2}\right)$$
\end{proof}

Note that in the proof above, $c \in (0, 1/2)$ is an arbitary constant which can be used to bound the probability of sampling an identity element from above. As $c \to 0, t^*(c) \to 2\pi$ and as $c \to 1/2, t^*(c) \to 0$. Since $t^*(c)$ is an invertible function of $c$ in this interval, the converse relation also holds. Consequently, we can arbitrarily fix $t^*(c) \in (0, 2\pi)$, and this will correspond to a given $c \in (0, 1/2)$. Since $t < \frac{t^*(c)}{2L}$, this means that we can evolve for any constant time $t \in (0, \pi/L)$. With this, we can now compile all the steps listed previously into the following protocol for deciding the intolerant emptiness testing problem.

\begin{algorithm}[H]
\caption{Intolerant Emptiness Testing}\label{alg:intol_empt}
\begin{algorithmic}[1]
\REQUIRE Query access to $U(t) = e^{-iHt}$, spectral norm bound $L$, emptiness threshold $\epsilon$, error rate $\delta$
\STATE Set $t \leftarrow \mathcal O(1/L)$.
\STATE Set $N \leftarrow \mathcal O\left(\frac{L^2\log(1/\delta)}{\epsilon^2}\right)$ 
\FOR{$i = 1, \dots, N$}
    \STATE Sample $\sigma_x$ from $U(t) = \sum_x U_x \sigma_x$ via Bell-sampling
    \IF{$\sigma_x \ne I$}
        \RETURN \texttt{NOT EMPTY}
    \ENDIF
\ENDFOR
\RETURN \texttt{EMPTY}
\ENSURE Whether $H$ is empty $(H = 0)$ or non-empty $(\lVert H \rVert_F \ge \epsilon)$ with success probability $\ge 1 - \delta$.
\end{algorithmic}
\end{algorithm}

With this algorithm, we can now distinguish between a Hamiltonian that is completely empty from one that has norm greater than $\epsilon$ in $\mathcal O(1/\epsilon^2)$ time. However, these criteria for an empty Hamiltonian are somewhat restrictive, as we still want to classify Hamiltonians with negligible norm as being empty. In order to resolve this issue, we generalize this algorithm to the tolerant framework.

\subsection{Tolerant Testing}

For tolerant emptiness testing, we instead have the following decision problem.

\begin{prob}[Tolerant Emptiness Testing]
    Given traceless Hamiltonian $H$ with bounded spectral norm $\lVert H \rVert_2 \le L$ that satisfies either $\lVert H \rVert_F \le \epsilon_1$ or $\lVert H \rVert_F \ge \epsilon_2$, where $\epsilon_2 > \epsilon_1$, decide whether $H$ is either $\epsilon_1$-close or $\epsilon_2$-far from zero.
\end{prob}

Note that intolerant emptiness testing problem arises as a special case of this problem where we set $\epsilon_1 = 0$. Consequently, instead of sampling Pauli operators until a non-identity Pauli is sampled, we instead need enough samples in order to distinguish between two Bernoulli distributions. This approach is formalized in the following theorem.

\begin{thm}\label{thm:tol_empt}
    There exists an algorithm that decides the tolerant emptiness testing problem for $H$ with probability of success $\ge 1 - \delta$ for a given threshold $\epsilon$ with total number of samples
    $$N = \mathcal O\left(\frac{L^2\log(1/\delta)}{\epsilon_2^2(1 - \epsilon_1^2/\epsilon_2^2)^3}\right)$$
    and total evolution time
    $$T = \mathcal O\left(\frac{L\log(1/\delta)}{\epsilon_2^2(1 - \epsilon_1^2/\epsilon_2^2)^{2.5}}\right)$$
\end{thm}

\begin{proof}
    From Theorem~\ref{thm:identity_bounds}, we have that
    $$1 - \lVert H \rVert_F^2t^2 \le \Pr[I] \le 1 - 2c\lVert H \rVert_F^2t^2$$
    To distinguish between the cases $\lVert H \rVert_F < \epsilon_1, \lVert H \rVert_F > \epsilon_2$, we must choose a value of $c$ such that these bounds do not overlap in both cases. 
    
    If $\lVert H \rVert_F < \epsilon_1$, we have that
    $$\Pr[I] \ge 1 - \epsilon_1^2 t^2 \implies \Pr[\lnot I] \le \epsilon_1^2t^2 = p_1$$
    Otherwise, if $\lVert H \rVert_F > \epsilon_2$, we have that
    $$\Pr[I] \le 1 - 2c\epsilon_2^2 t^2 \implies \Pr[\lnot I] \ge 2c\epsilon_2^2t^2 = p_2$$

    Therefore, we must choose $c$ s.t.
    $$1 - 2c\epsilon_2^2 t^2 < 1 - \epsilon_1^2 t^2 \implies c > \frac{\epsilon_1^2}{2\epsilon_2^2}$$
    Since $\epsilon_1 < \epsilon_2$, we have that $\frac{\epsilon_1^2}{2\epsilon_2^2} < 1/2$, meaning that there is always a feasible range of $c$ allowed:
    $$\frac{\epsilon_1^2}{2\epsilon_2^2} < c < \frac{1}{2}$$

    Under this constraint, we have that $2c\epsilon_2^2t^2 > \epsilon_1^2t^2$, meaning that we can use Monte Carlo estimation to distinguish between these two cases with $N \sim 1/\epsilon_2^2$ samples for $\epsilon_1 = \mathcal O(\epsilon_2)$ \cite{PRXQuantum.3.040305}. Essentially, we simply estimate the probability of sampling an identity element and then check whether the result computed is less than or greater than $p_{1/2} =\frac{p_1 + p_2}{2}$. Using the Chernoff-Hoeffding bound, we have that the probability of incorrectly choosing the distribution can be upper bounded as follows:
    $$\Pr\left(\text{error}\right) \le \max\{e^{-D(p_{1/2}||p_1)N},e^{-D(p_{1/2}||p_2)N}\}$$
    To upper bound this quantity, we simply need to lower bound the Kullback-Leibler Divergence expressions. In particular, we use the lower bound stated in Lemma \ref{lem:kl} to get the following:
    $$\min\left\{D(p_{1/2}||p_1), D(p_{1/2}||p_2)\right\} \ge \min\left\{\frac{(p_{1/2} - p_1)^2}{2p_{1/2}}, \frac{(p_2 - p_{1/2})^2}{2p_{2}}\right\} = \frac{(p_2 - p_1)^2}{8p_{2}} = \frac{(2c\epsilon_2^2 - \epsilon_1^2)^2t^2}{16c\epsilon_2^2}$$
    Consequently, to obtain an error rate $\le \delta$, we can set
    $$\Pr\left(\text{error}\right) \le \max\{e^{-D(p_{1/2}||p_1)N},e^{-D(p_{1/2}||p_2)N}\} \le \exp\left(-\frac{N(2c\epsilon_2^2 - \epsilon_1^2)^2t^2}{16c\epsilon_2^2}\right) \le \delta$$
    Rearranging gives us that
    $$N \ge \frac{16c\epsilon_2^2\log(1/\delta)}{(2c\epsilon_2^2 - \epsilon_1^2)^2t^2}$$
    Consequently, we have that a sufficient sampling complexity is on the following order
    $$N = \mathcal O\left(\frac{\epsilon_2^2\log(1/\delta)}{(2c\epsilon_2^2 - \epsilon_1^2)^2t^2}\right)$$
    
    For any value of $c$, we can approximate the corresponding optimal value of $t$ using the most restrictive range for $t$, which occurs when $|\lambda_j - \lambda_k| = 2\lVert H \rVert_2 \le 2L$, giving us the following:
    $$\cos(2Lt) = 1 - 4cL^2t^2 \implies c = \frac{1 - \cos(2Lt)}{4L^2t^2} = \frac{\sin^2(Lt)}{2L^2t^2} \approx \frac12 - \frac{L^2t^2}{6}$$
    Rearranging for $t$ therefore gives us
    $$t \approx \frac{\sqrt{3 - 6c}}{L}$$
    Substituting this into the expression from before gives us
    $$N = \mathcal O\left(\frac{L^2\epsilon_2^2\log(1/\delta)}{(2c\epsilon_2^2 - \epsilon_1^2)^2(3 - 6c)}\right)$$
    Maximizing the denominator by setting the derivative with respect to $c$ equal to zero, we find the optimal $c$ to be $c = \frac{\epsilon_1^2 + \epsilon_2^2}{4\epsilon_2^2}$, or in other words, the arithmetic mean of the upper and lower bounds on $c$. Substituting this in gives us
    $$N = \mathcal O\left(\frac{L^2\log(1/\delta)}{\epsilon_2^2(1 - \epsilon_1^2/\epsilon_2^2)^3}\right)$$
    For the total evolution time $T$, we multiply by an additional factor of $t \approx \frac{\sqrt{3 - 6c}}{L} = \frac{\sqrt{\frac{3}{2}(1 - \epsilon_1^2/\epsilon_2^2)}}{L}$, giving us that
    $$T = \mathcal O\left(\frac{L\log(1/\delta)}{\epsilon_2^2(1 - \epsilon_1^2/\epsilon_2^2)^{2.5}}\right)$$
\end{proof}

With the choice of $c = \frac{\epsilon_1^2 + \epsilon_2^2}{4\epsilon_2^2}$, we have that
$$p_{1/2} = \frac{\left(\epsilon_1^2 + 2c\epsilon_2^2\right)t^2}{2} = \frac{(3\epsilon_1^2 + \epsilon_2^2)t^2}{4}$$

We can observe that this expression for the total runtime and sample complexity is consistent with that of Theorem~\ref{thm:intol_empt}, as substituting $\epsilon_1 = 0, \epsilon_2 = \epsilon$ yields the same result. Writing out the full algorithm for this procedure, we have the following:
\begin{algorithm}[H]
\caption{Tolerant Emptiness Testing}\label{alg:tol_empt}
\begin{algorithmic}[1]
\REQUIRE Query access to $U(t) = e^{-iHt}$, spectral bound $L$, emptiness thresholds $\epsilon_1, \epsilon_2$, error rate $\delta$
\STATE Set $t \leftarrow \mathcal O\left(\frac{\sqrt{1 - \epsilon_1^2/\epsilon_2^2}}{L}\right)$.
\STATE Set $N \leftarrow \mathcal O\left(\frac{L^2\log(1/\delta)}{\epsilon^2(1 - \epsilon_1^2/\epsilon_2^2)^3}\right)$ 
\STATE Initialize $m \leftarrow 0$
\FOR{$i = 1, \dots, N$}
    \STATE Sample $\sigma_x$ from $U(t) = \sum_x U_x \sigma_x$ via Bell-sampling
    \IF{$\sigma_x \ne I$}
        \STATE Increment $m \leftarrow m + 1$
    \ENDIF
\ENDFOR
\IF{$m/N \ge \frac{(3\epsilon_1^2 + \epsilon_2^2)t^2}{4}$}
    \RETURN \texttt{NOT EMPTY}
\ELSE
    \RETURN \texttt{EMPTY}
\ENDIF
\ENSURE Whether $H$ is empty $(\lVert H \rVert_F < \epsilon_1)$ or non-empty $(\lVert H \rVert_F > \epsilon_2)$ with success probability $\ge 1 - \delta$.
\end{algorithmic}
\end{algorithm}

\section{Sparse Hamiltonian Learning}

Learning sparse Hamiltonians in the intolerant regime has previously been shown to be done with Heisenberg scaling \cite{hu2025ansatzfreehamiltonianlearningheisenberglimited}. We present a refined average-case evolution-time analysis of the algorithm in \cite{hu2025ansatzfreehamiltonianlearningheisenberglimited}.

\begin{prob}[Sparse Hamiltonian Learning] 
    Given an exactly $M$-sparse Hamiltonian $H = \sum_x \mu_x P_x$ with unknown Pauli coefficients where $|\mu_x| \le 1$, learn a $M$-sparse approximation $\hat H = \sum_x \hat \mu_x \sigma_x$ such that $\lVert \boldsymbol \mu - \boldsymbol {\hat \mu} \rVert_\infty \le \epsilon$ where $\boldsymbol \mu = (\mu_1, \mu_2, \dots), \boldsymbol {\hat \mu} = (\hat \mu_1, \hat \mu_2, \dots)$. 
\end{prob}

\begin{thm}[Sparsity Learning\cite{hu2025ansatzfreehamiltonianlearningheisenberglimited}]\label{thm:intol_learn}
    There exists an algorithm that given a $M$-sparse Hamiltonian $H = \sum_x \mu_xP_x$, can compute a $M$-sparse approximation $\hat H = \sum_x \hat \mu_x P_x$ such that
    $$\lVert \boldsymbol \mu - \boldsymbol {\hat \mu} \rVert_\infty \le \epsilon$$
    where $\boldsymbol \mu = (\mu_1, \mu_2, \dots), \boldsymbol {\hat \mu} = (\hat \mu_1, \hat \mu_2, \dots)$ with high probability and expected total evolution time
    $$T = \widetilde {\mathcal O}\left(\frac{M}{\epsilon}\right)$$
\end{thm}

Here, we use $\widetilde {\mathcal O}(\cdot)$ to omit logarithmic factors for brevity. Compared to the result in \cite{hu2025ansatzfreehamiltonianlearningheisenberglimited}, we present an improved analysis which removes a factor of $M$ from the average evolution time required, giving us the result shown above.

This learning algorithm consists of two subroutines: a protocol to learn the support of the Hamiltonian (denoted $\mathcal A^{I}$) and a protocol to learn the coefficients of Pauli operators in the support of the Hamiltonian (denoted $\mathcal A^{II}$). This algorithm works in hierarchical fashion, grouping Pauli coefficients into buckets $\mathfrak S_j$ defined as follows for $j = 0, \dots, \lceil \log_2(1/\epsilon)\rceil - 1$:
$$\mathfrak S_j = \{x : 2^{-(j + 1)} < |\mu_x| \le 2^{-j}\}$$
In particular, for each $\mathfrak S_j$, we apply $\mathcal A^I$ and $\mathcal A^{II}$ to the Hamiltonian $\tilde H_j = \frac{H - \hat H_j}{2^{-j}}$, where $\hat H_j$ denotes the Hamiltonian learned up to this point. By using this Hamiltonian, we kill off the terms we have already learned in $H$ and boost the remaining terms by a factor of $2^{j}$. For $x \in \mathfrak S_j$, we have that $2^{-(j + 1)} < |\mu_x| \le 2^{-j}$, meaning that $2^j|\mu_x|$ is of constant order and can be learned easily from $\mathcal A^I$ and $\mathcal A^{II}$.

To learn the support of the Hamiltonian in $\mathcal A^I$, we simply use Bell sampling to determine the Pauli operators we can sample. In particular, we time-evolve for $t = \mathcal O\left(\frac{1}{CM}\right)$, which for any $x \in \mathfrak S_j$ gives us a sampling probability of
$$\gamma_j = \tilde\Omega\left(\frac{1}{4C^2M^2} - \frac{1}{C^3M^2}\right),$$
as stated in Appendix C of \cite{hu2025ansatzfreehamiltonianlearningheisenberglimited}. In order to sample all the indices in $\mathfrak S_j$, we then need to sample
$$\mathcal O\left(\frac{\log(M/\delta)}{\gamma_j}\right) = \mathcal O(M^2 \log (M/\delta))$$
times. Consequently, in the maximal case, we will sample $\mathcal O(M^2 \log M)$ different Pauli operators that could be in the support of $H$. Altogether, this subroutine is given as follows.

\begin{algorithm}[H]
\caption{Subroutine: Structure Learning Algorithm $\mathcal A^I$ \cite{hu2025ansatzfreehamiltonianlearningheisenberglimited}}\label{alg:intol_sparse_learn_i}
\begin{algorithmic}[1]
\REQUIRE Query access to $U(t) = e^{-iHt}$, approximate learned Hamiltonian $\hat H$, iteration $j$, sparsity $M$, error rate $\delta$
\STATE Initialize $\mathcal B \leftarrow \varnothing$
\STATE Set evolution time $\tau \leftarrow \mathcal O\left(\frac{2^j}{M}\right)$
\STATE Set trotterization time steps $r_1 \leftarrow \mathcal O\left(2^{2j}M^2\right)$
\STATE Set the total number of measurements $\mathfrak M = \mathcal O(M^2\log (M/\delta))$
\FOR{$i = 1, \dots, \mathfrak M$}
    \STATE Perform Bell-sampling on $(e^{i\hat H\tau/r_1}U(\tau/r_1))^{r_1}$ to get outcome $b_i$
    \STATE If $P \ne I$, set $\mathcal B \leftarrow \mathcal B \cup \{b_i\}$
\ENDFOR
\RETURN $\mathcal B$
\ENSURE Obtain estimate $\mathcal B$ of the support set $\mathfrak S_j$
\end{algorithmic}
\end{algorithm}

Once we have confidently sampled all the indices in $\mathfrak S_j$ (potentially oversampling in the process), we then learn the coefficients with $\mathcal A^{II}$. To do this, we first use Hamiltonian reshaping \cite{HuangTongFangSu2023learning,ma2024learningkbodyhamiltonianscompressed} in order to isolate the Pauli in the Hamiltonian whose coefficient we aim to learn. In particular, suppose we aim to learn Pauli $P_s$. Define $\mathcal K_{P_s} = \{P \in \mathbb P_n: [P_s, P] = 0\}$. If we now perform Pauli twirling using Pauli operators in $\mathcal K_{P_s}$, using the fact that for any Pauli $P \ne P_s,$ $P$ commutes with half of the Pauli operators in $\mathcal K_{P_s}$ and anticommutes with rest, we get the following single-step quantum channel:
$$\rho \mapsto \frac{1}{2^{2n - 1}}\sum_{Q \in \mathcal K_{P_s}} Qe^{-iHt}Q\rho Qe^{iHt}Q = \rho - i\tau [H_\text{eff}, \rho] + \mathcal O(\tau^2)$$
where the effective Hamiltonian is given as
$$H_\text{eff} = \frac{1}{2^{2n - 1}}\sum_{Q \in \mathcal K_{P_s}} QHQ = P_s$$

Consequently, with this procedure, we effectively time-evolve with respect to a single term $P_s$ in the Hamiltonian. By applying this randomized reshaping channel over $r_2$ small time-steps, we can approximate evolution under $e^{-i\mu_sP_st}$ with diamond norm error scaling $\mathcal O(M^2t^2/r_2)$ \cite{hu2025ansatzfreehamiltonianlearningheisenberglimited}. The coefficient $\mu_s$ is then estimated using robust frequency estimation \cite{ma2024learningkbodyhamiltonianscompressed, PhysRevA.92.062315}. In particular, if we want to learn the coefficient of $P_s = \otimes_{k = 1}^n\sigma_{\beta_{s,k}}$, we first prepare the state
$$|\phi_0^s\rangle = \frac{1}{\sqrt 2}(|1, \beta_s\rangle + |-1, \beta_s\rangle) \quad \left(|\pm 1, \beta_s\rangle = \bigotimes_{k = 1}^n|\pm1, \beta_{s, k}\rangle\right)$$
and $|\pm 1, \beta_{s, k}\rangle$ is the $\pm 1$ eigenstate of $\sigma_{\beta_{s, k}}$. We then evolve this state using Hamiltoniann reshaping to get
$$|\phi_t^s\rangle = \left(\prod_{i = 1} ^{r_2}Q_iU(\tau/r_2)Q_i\right)|\phi_0^s\rangle = \frac{1}{\sqrt 2}(e^{-i\mu_st}|1, \beta_s\rangle + e^{i\mu_st}|-1, \beta_s\rangle)$$
We define the following observables:
\begin{align*}
    O_s^+ &= \left( \bigotimes_{j=1}^{s^\star - 1} \left|1, \beta_{s,j}\right\rangle \right) \otimes Q_{s,s^\star}^+ \otimes \left( \bigotimes_{j=s^\star + 1}^n \left|1, \beta_{s,j}\right\rangle \right), \\
    O_s^- &= \left( \bigotimes_{j=1}^{s^\star - 1} \left|1, \beta_{s,j}\right\rangle \right) \otimes Q_{s,s^\star}^- \otimes \left( \bigotimes_{j=s^\star + 1}^n \left|1, \beta_{s,j}\right\rangle \right),
\end{align*}
where $Q_{s,s^\star}^+$ and $Q_{s,s^\star}^-$ are chosen to be single-qubit Pauli operators such that
\begin{align*}
    Q_{s,s^\star}^+ \left|1, \beta_{s,s^\star}\right\rangle &= \left|-1, \beta_{s,s^\star}\right\rangle, & Q_{s,s^\star}^+ \left|-1, \beta_{s,s^\star}\right\rangle &= \left|1, \beta_{s,s^\star}\right\rangle; \\
    Q_{s,s^\star}^- \left|1, \beta_{s,s^\star}\right\rangle &= i\left|-1, \beta_{s,s^\star}\right\rangle, & Q_{s,s^\star}^- \left|-1, \beta_{s,s^\star}\right\rangle &= -i\left|1, \beta_{s,s^\star}\right\rangle.
\end{align*}
We then measure either $O_s^+$ or $O_s^-$ by measuring $s^\star$th qubit, repeating this procedure until we can estimate
$$S = \langle O_s^+\rangle + i\langle O_s^-\rangle = e^{-i\frac{2\mu_s\pi}{b-a}}$$
to a certain constant precision that is independent of the problem (independent of $\epsilon$ in particular).
We then accordingly update our definition of $a, b$ in order to compute $\hat \mu_s$, as shown in the algorithm below.

\begin{algorithm}[H]
\caption{Subroutine: Coefficient Learning Algorithm $\mathcal A^{II}$ \cite{hu2025ansatzfreehamiltonianlearningheisenberglimited}}\label{alg:intol_sparse_learn_ii}
\begin{algorithmic}[1]
\REQUIRE Query access to $U(t) = e^{-iHt}$, Pauli $P_s$, learning accuracy $\epsilon$, sparsity $M$
\STATE Set $a \leftarrow -\pi, b \leftarrow \pi$
\STATE Set $N_\text{exp} = \widetilde{\mathcal O}(1)$
\FOR{$\ell = 1, \dots, \log_{3/2}(2\pi/\epsilon)$}
    \STATE Prepare the initial state $|\phi_0^s\rangle$
    \STATE Set evolution time $\tau \leftarrow \frac{\pi}{2(b - a)}$
    \STATE Set Trotterization steps $r_2 \leftarrow \mathcal O(M^2\tau^2)$
    \STATE Evolve $|\phi_0^t\rangle = \left(\prod_{i = 1} ^{r_2}Q_iU(\tau/r_2)Q_i\right)|\phi_0^s\rangle$, where $Q_i \in \mathcal K_{P_s}$ are sampled uniformly
    \STATE Measure either $O_s^+$ or $O_s^-$ on $|\psi_-'\rangle$ each $N_\text{exp}$ times and compute estimate $\hat S = \langle O_s^+\rangle + i\langle O_s^-\rangle$ of $\exp\left(-i\frac{2\mu_s \pi}{b - a}\right)$
    \IF{$\text{Im}\left[\exp\left(-i\frac{(a + b)\pi}{b - a}\hat S\right)\right] \le 0$}
        \STATE $b \leftarrow (a + 2b)/3$
    \ELSE
        \STATE $a \leftarrow (2a + b)/3$
    \ENDIF
\ENDFOR
\RETURN $\hat \mu_s = (b - a)/4$
\ENSURE Estimate $\hat \mu_s$ of coefficient $\mu_s$ for Pauli $P_s$ where $|\hat \mu_s - \mu_s| \le \epsilon$
\end{algorithmic}
\end{algorithm}

We then combine these two subroutines in hierarhical fashion, learning the support and then the coefficients for Pauli operators with coefficients within certain buckets in order to achieve Heisenberg-limited total evolution time.

\begin{algorithm}[H]
\caption{Intolerant Sparsity Learning \cite{hu2025ansatzfreehamiltonianlearningheisenberglimited}}\label{alg:intol_sparse_learn}
\begin{algorithmic}[1]
\REQUIRE Query access to $U(t) = e^{-iHt}$, sparsity $M$, learning accuracy $\epsilon$, error rate $\delta$
\STATE Initialize all learned Pauli coefficients $\hat \mu_x = 0$ and define $\hat H = \sum_{x} \hat \mu_xP_x$
\FOR{$j = 0, \dots, \lceil \log_2(1/\epsilon)\rceil - 1$}
    \STATE Run Algorithm \ref{alg:intol_sparse_learn_i} on $U(t), \hat H, j, M, \delta$ and store output $\to \mathcal B$
    \FOR{$s \in \mathcal B$}
        \STATE Run Algorithm \ref{alg:intol_sparse_learn_ii} on $U(t), P_s, \epsilon, M$ and store output $\to \hat \mu_s$
    \ENDFOR
\ENDFOR
\RETURN $\hat H$
\ENSURE $M$-sparse approximation $\hat H = \sum_{x} \hat \mu_x P_x$ of $H$ where $\lVert \hat {\boldsymbol \mu} - \boldsymbol \mu \rVert_\infty \le \epsilon$ with high probability if $H$ is $M$-sparse.
\end{algorithmic}
\end{algorithm}

In the original paper \cite{hu2025ansatzfreehamiltonianlearningheisenberglimited}, the total evolution times $T_1$ for $\mathcal A^I$ and $T_2$ for $\mathcal A^{II}$ are given as follows:
$$T_1 = \widetilde{\mathcal O}(M/\epsilon), \quad T_2 = \widetilde{\mathcal O}(M^2/\epsilon)$$
Consequently, the runtime is dominated by $T_2$. However, this expression is a worst-case scenario in which we obtain $\mathcal O\left(M^2\log M\right)$ unique Pauli operators whose coefficients we have to learn from $\mathcal A^{I}$. We improve upon this analysis by showing that we only sample $\mathcal O(M \log M)$ non-identity Pauli operators in each iteration of $\mathcal A^I$.

\begin{lem}[Non-Identity Sampling Complexity]\label{lem:learning_sample}
    For the Hamiltonian $\tilde H_j = \frac{H - \hat H_j}{2^{-j}}$ we time-evolve with respect to in each round of the support learning algorithm $\mathcal A^I$, performing Bell-sampling $\Theta(M^2 \log M)$ times yields $\le \Theta(\alpha M \log M)$ non-identity Pauli operators w.p. $\ge 1 - e^{-\Theta(\alpha M \log M)}$ for sufficiently large $\alpha > 1$.
\end{lem}

\begin{proof}
Using the lower bound from Theorem \ref{thm:identity_bounds}, we have that
$$\Pr[I] \ge 1 - \lVert H \rVert_F^2t^2$$
We define $S_j:=\{x:|\mu_x|\leq 2^{-j}\}$. Computing an upper bound on the normalized Frobenius norm of $\tilde H_j$, we have the following:
$$\lVert \tilde H_j \rVert_F = 2^j\sqrt{\sum_{x \in S}(\mu_x - \hat \mu_x)^2} = 2^j\sqrt{\sum_{x \in S_0\setminus S_j}(\mu_x - \hat \mu_x)^2 + \sum_{x \in S_j}\mu_x^2} \le 2^j\sqrt{\varepsilon^2 (|S_0| - |S_j|) + 2^{-2j}|S_j|} \le \sqrt{M}$$
The last inequality follows from the fact that $2^{-j} \ge \epsilon$ since $j \le \lceil \log_2(1/\epsilon) \rceil - 1$. Therefore, choosing $t = \Theta\left(\frac1M\right)$, we have that
$$\Pr[I] \ge 1 - \lVert \tilde H_j \rVert_F^2t^2 \ge 1 - \Theta\left(\frac{1}{M}\right)$$
In other words, the probability of sampling a non-identity element is bounded above as
$$\Pr[\lnot I] \le \Theta\left(\frac{1}{M}\right)$$
Consequently, since we are sampling $N = \Theta(M^2 \log M)$ times, the expected number of non-identity Pauli operators we sample is bounded above by $\Theta(M \log M)$.

We now derive a concentration inequality  for the probability of sampling more than $\mathcal \Theta(\alpha M \log M)$ for some constant $\alpha > 1$. Let $X_i \sim \text{Bernoulli}(\Theta(1/M))$ denote a given sampling event where $X_i = 1$ denotes a non-identity Pauli and $X_i = 0$ denotes an identity Pauli. Consequently, we can define the random variable corresponding to the total number of non-identity Pauli operators sampled as
$$X = \sum_{i = 1} ^{N} X_i$$
By the multiplicative Chernoff bound, we have that
$$\Pr[X \ge \alpha \mathbb E[X]] \le \exp\left(-\frac{(\alpha - 1)^2\mathbb E[X]}{1 + \alpha}\right) \le \exp\left(-\Theta(\alpha \mathbb E[X])\right),$$
where $\mathbb E[X] = N/M = \Theta(M \log M)$.

Altogether, this means that we sample less than $\Theta(\alpha M \log M)$ non-identity Pauli operators with probability
$$\Pr[X \le \alpha \cdot \Theta(M \log M)] \ge 1 - e^{-\Theta(\alpha M \log M)}$$

\end{proof}

Since the probability of sampling more than $\Theta(\alpha M \log M)$ Pauli operators decays exponentially in $\alpha$, on average, we have that the total-evolution time required for learning step and therefore the total learning protocol is given as
$$T = \widetilde{\mathcal O}\left(\frac{M}{\epsilon}\right)$$
Consequently, with Lemma \ref{lem:learning_sample}, we improve the average runtime of the algorithm from $\widetilde{\mathcal O}(M^2/\epsilon)$ to $\widetilde{\mathcal O}(M/\epsilon)$, as stated in Theorem \ref{thm:intol_learn}.





\section{Sparsity Testing}

\begin{prob}[Intolerant Sparsity Testing]
    Given Hamiltonian $H = H' + H_\Delta$ with bounded spectral norm $\lVert H \rVert_2 \le L$ for $L > 0$, where $H'$ is the closest $M$-sparse approximation to $H$ in the Frobenius norm, determine if $H$ is exactly $M$-sparse ($H_\Delta = 0$) or $\epsilon$-far from $M$-sparse ($\lVert H_\Delta \rVert_F \ge \epsilon$).
\end{prob}

\begin{lem}\label{lem:trotter}
    Suppose we are given Hamiltonians $H, \hat H$. Define the a single step of the Trotterized time-evolution operator as
    $$U(t) = e^{i\hat Ht}e^{-iHt}$$
    We consider a bipartite system initialized in the maximally entangled state $\ket{\Phi}$.
    If we have to time-evolve a given Hamiltonian for $r$ time-steps for total time $T$ in one part of the bipartite system, we can bound the total additive Trotter error between $U(T/r)^r$ and $e^{-i(H - \hat H)T}$ as
    $$\varepsilon_\text{Trotter} :=  \left\lVert ((U(t)^r \otimes I) - (e^{-i(H - \hat H)rt} \otimes I))|\Phi\rangle \right\rVert \le \frac{T^2\lVert \hat H \rVert_F\lVert H \rVert_2}{r},$$
    where $t = T/r$ is the duration of the time-step.
\end{lem}

\begin{proof}
    We can observe that
    \begin{align*}
        \varepsilon_\text{Trotter} &\le \lVert ((U(t)^r -U(t)e^{-i(H - \hat H)(r - 1)t}) \otimes I)|\Phi\rangle \rVert + \lVert ((U(t)e^{-i(H - \hat H)(r - 1)t} - e^{-i(H - \hat H)rt}) \otimes I)|\Phi\rangle \rVert \\
        &= \lVert ((U(t)^{r - 1} - e^{-i(H - \hat H)(r - 1)t}) \otimes I)|\Phi\rangle \rVert + \lVert U(t) - e^{-i(H - \hat H)t}\rVert_F
    \end{align*}
    The RHS is simplified using Fact \ref{fact:bell_norm} concerning the maximally entangled state and the invariance of the Frobenius norm under unitary operations. We can observe that the first term is identical to the LHS with the substitution $r \to r - 1$, giving us a recursive bound on the norm. Since the offset term is independent of $r$, we have that
    $$\varepsilon_\text{Trotter} \le r\lVert U(t) - e^{-i(H - \hat H)t}\rVert_F$$
    
    The expression in the Frobenius norm is simply the additive error in the Lie-Trotter Formula, which from \cite{PhysRevX.11.011020} is given as
    $$\mathscr A(t) = U(t) - e^{-i(H - \hat H)t} = \int_0 ^t d\tau \; e^{-i(t - \tau)(H - \hat H)}[e^{i\tau \hat H}, -iH]e^{-i\tau H}$$
    Computing the error, we have that since we are evolving against a Bell state, this error is given as follows:
    $$\lVert \mathscr A(t) \rVert_F \le \int_0 ^t d\tau \; \lVert e^{-i(t - \tau)(H - \hat H)}[e^{i\tau \hat H}, -iH]e^{-i\tau H} \rVert_F = \int_0 ^t d\tau \; \lVert [e^{i\tau \hat H}, -iH] \rVert_F$$
    Since the left and right operators are unitary, they leave the Frobenius norm invariant and can be ignored. To bound the norm of the commutator, we have the following:
    $$\lVert [e^{i\tau \hat H}, H] \rVert_F = \lVert [e^{i\tau \hat H} - I, H]\rVert_F \le 2\lVert e^{i\tau \hat H} - I\rVert_F\lVert H \rVert_2$$

    Using the Taylor integral remainder formula, we get the following:
    $$\lVert e^{i\tau \hat H} - I \rVert_F = \left\lVert \int_{0} ^\tau dt' \; i\hat He^{it' \hat H}\right\rVert_F \le \int_{0} ^\tau dt' \; \lVert\hat H\rVert_F \left\lVert e^{it' \hat H}\right\rVert_2 \le \tau \lVert \hat H \rVert_F$$
    Consequently, this gives us that
    $$\lVert [e^{i\tau \hat H}, H] \rVert_F \le 2\tau\lVert \hat H \rVert_F\lVert H \rVert_2$$

    Therefore, we have that
    $$\lVert \mathscr A(t) \rVert_F \le \int_0 ^t d\tau \; 2\tau\lVert \hat H \rVert_F\lVert H \rVert_2 = t^2\lVert \hat H \rVert_F\lVert H \rVert_2$$

    Substituting this back into the bound for $\varepsilon_\text{Trotter}$, we have that
    $$\varepsilon_\text{Trotter} \le rt^2\lVert \hat H \rVert_F\lVert H \rVert_2 = \frac{T^2\lVert \hat H \rVert_F\lVert H \rVert_2}{r}$$
\end{proof}

\begin{thm}\label{thm:intol_sparse}
    There exists an algorithm that decides the intolerant $M$-sparsity testing problem for $H$ (where $\lVert H \rVert_2 \le L$) with probability of success $\ge 1 - \delta$ for a given threshold $\epsilon$ with total evolution time
    $$T = \widetilde {\mathcal O}\left(\frac{LM^{1.5}}{\epsilon} + \frac{LM}{\epsilon^2}\right)$$
\end{thm}

\begin{proof}
    Since the learning algorithm assumes that for $H = \sum_x \mu_x P_x$, $|\mu_x| \le 1$, but we are only guaranteed that $\lVert H \rVert_2 \le L$, we must rescale $H$ by a factor of $1/L$ for the learning algorithm. Doing so, gives us a $M$-sparse approxiation $\hat H/L$, which we then rescale accordingly by a factor of $L$ to get the true Hamiltonian.
    
    Suppose the Hamiltonian $H$ is exactly $M$-sparse. Directly applying the intolerant sparsity learning algorithm, we can learn each coefficient to precision $\frac{\epsilon}{2L\sqrt M}$ with high probability, where
    $$\lVert H - \hat H \rVert_F \le \frac{\epsilon}{2}$$
    On the other hand, if $H$ is $\epsilon$-far from $M$-sparse, any $M$-sparse approximation $\hat H$ to $H$ will be at least $\epsilon$-far from $\hat H$, giving us that
    $$\lVert H - \hat H \rVert_F \ge \epsilon$$
    Consequently, if we now apply Theorem~\ref{thm:tol_empt}, we can use the tolerant emptiness testing protocol to distinguish between these two cases with probability $\ge 1 - \delta$. In particular, we have that $\epsilon_1 = \epsilon/2, \epsilon_2 = \epsilon$.

    In order to compute the time $t$ we must evolve for during each iteration of Bell-sampling, we need to bound $\lVert H - \hat H \rVert_2 \le \lVert H \rVert_2 + \lVert \hat H \rVert_2 \le L + LM$. Here, the upper bound on $\lVert \hat H \rVert_2$ arises from the fact that the initial sparse Hamiltonian produced directly from the learning algorithm has at most $M$ Pauli operators, each with coefficients of absolute value less than 1. After rescaling by $L$, we get that $\hat H$ has spectral norm upper bounded by $LM$. Therefore, we have that $t = \mathcal O(1/(LM))$. To time-evolve with respect to $H - \hat H$, we simply use Trotterization to time-evolve for time $\tau$ with respect to $H$ and then reverse time-evolution under $\hat H$ for time $\tau$. By Lemma \ref{lem:trotter}, using the fact that $\lVert \hat H \rVert_F \le L\sqrt M, \lVert H \rVert_2 \le L$, we get that
    $$\varepsilon_\text{Trotter} \le \frac{t^2L^2\sqrt M}{r}$$
    In order to determine a sufficient bound on the Trotter error, we must bound the absolute value of the amplitude of the component of the state orthogonal to $|\Phi\rangle$. Defining $U$ as the ideal time evolution operator, $U'$ as the Trotterized approximation, and $\Pi = I - |\Phi\rangle \langle \Phi|$ we have the following lower bound on the amplitude corrresponding to the component orthogonal to $|\Phi\rangle$:
    $$\lVert \Pi U'|\Phi\rangle\rVert_2 \ge \left|\lVert \Pi U|\Phi\rangle\rVert_2 - \lVert \Pi (U' - U)|\Phi\rangle\rVert_2\right| \ge \lVert \Pi U|\Phi\rangle\rVert_2 - \varepsilon_\text{Trotter}$$
    By using the other direction of the triangle inequality, we can also obtain an equivalent upper bound
    $$\lVert \Pi U'|\Phi\rangle\rVert_2 \le \lVert \Pi U|\Phi\rangle\rVert_2 + \lVert \Pi (U' - U)|\Phi\rangle\rVert_2 \le \lVert \Pi U|\Phi\rangle\rVert_2 + \varepsilon_\text{Trotter}$$
    Consequently, we have that in the worst-case, the amplitude is changed by at most $\varepsilon_\text{Trotter}$, irrespective of the original amplitude. Therefore, we only need to consider the worst-case changes, where the worst-case amplitude in the sparse case is increased by $\varepsilon_\text{Trotter}$ and the worst-case amplitude in the non-sparse case is decreased by $\varepsilon_\text{Trotter}$. For our purposes, we choose
    $$\varepsilon_\text{Trotter} \le \frac{\epsilon t}{8},$$
    giving us that
    $$r = \frac{8tL^2\sqrt{M}}{\epsilon}$$
    In the sparse case, we have that the largest possible sampling probability in the ideal setting is $p_1 = \epsilon_1^2t^2 = \frac{\epsilon^2 t^2}{4}$, meaning that
    $$\lVert \Pi U'|\Phi\rangle\rVert_2 \le \epsilon_1't = \frac{5}{8}\epsilon t \implies \epsilon_1' = \frac{5\epsilon}{8}$$
    On the other hand, for the sparse case, we have that the smallest possible sampling probability in the ideal setting is $p_2 = 2c\epsilon_2^2t^2 =  \frac{(\epsilon_1^2 + \epsilon_2^2)t^2}{2} = \frac{5}{8}\epsilon^2t^2$, meaning that
    $$\lVert \Pi U'|\Phi\rangle\rVert_2 \ge \sqrt{2c'}\epsilon_2't = t\sqrt\frac{\epsilon_1'^2 + \epsilon_2'^2}{2} = \left(\sqrt\frac{5}{8} - \frac{1}{8}\right)\epsilon t \implies \epsilon_2' = \frac{\epsilon\sqrt{57-8\sqrt{10}}}{8}$$
    With this choice of $\varepsilon_\text{Trotter}$, we therefore have $\epsilon_2' > \epsilon_1'$. Consequently, we can use these as the corrected emptiness thresholds to account for the Trotterization error.
    
    Computing the total evolution time, we have from Theorem \ref{thm:intol_learn}, that the total evolution time for the learning algorithm with precision $\frac{\epsilon}{2L\sqrt M}$ is given as
    $$T_1 = \widetilde{\mathcal O}\left(\frac{LM^{1.5}}{\epsilon}\right)$$
    Furthermore, from Theorem~\ref{thm:tol_empt}, since we are testing emptiness with $\epsilon_1 = \epsilon/2$ and $\epsilon_2 = \epsilon$, we have a total evolution time of
    $$T_2 = \mathcal O\left(\frac{LM\log(1/\delta)}{\epsilon^2}\right)$$
    Combining these expressions and ignoring logarithmic factors gives a total evolution time of
    $$T = \widetilde {\mathcal O}\left(\frac{LM^{1.5}}{\epsilon} + \frac{LM}{\epsilon^2}\right)$$
\end{proof}

Altogether, we can write this algorithm as follows:
\begin{algorithm}[H]
\caption{Intolerant Sparsity Testing}\label{alg:intol_sparse_test}
\begin{algorithmic}[1]
\REQUIRE Query access to $U_H(t) = e^{-iHt}$, sparsity threshold $\epsilon$, sparsity $M$, spectral bound $L$, error rate $\delta$
\STATE Set Trotter parameters $t \leftarrow \mathcal O\left(\frac{1}{LM}\right), r \leftarrow \frac{8tL^2\sqrt M}{\epsilon}$
\STATE Set emptiness thresholds $\epsilon_1' \leftarrow \frac{5\epsilon}{8}, \epsilon_2' \leftarrow \frac{\epsilon\sqrt{57-8\sqrt{10}}}{8}$
\STATE Run Algorithm \ref{alg:intol_sparse_learn} with $U_1(t) = U_H(t/L)$, threshold $\frac{\epsilon}{2L\sqrt{M}}$ and store output scaled by $L$ $\to \hat H$.
\STATE Run Algorithm \ref{alg:tol_empt} with $U_2(t) = (e^{i\hat Ht/r}U_H(t/r))^r$, spectral bound $L(M + 1)$, $\epsilon_1 = \epsilon_1', \epsilon_2 = \epsilon_2', \delta$ and store output $\to A$.
\IF{$A = \texttt{EMPTY}$}
    \RETURN \texttt{$M$-SPARSE}
\ELSE
    \RETURN \texttt{NOT $M$-SPARSE}
\ENDIF
\ENSURE Whether $H$ is $M$-sparse or not with probability of success $\ge 1 - \delta$
\end{algorithmic}
\end{algorithm}

\bibliographystyle{unsrt}
\bibliography{refs}

\begin{thebibliography}{10}

\bibitem{Seif2021compressed}
Alireza Seif, Mohammad Hafezi, and Yi-Kai Liu.
\newblock Compressed sensing measurement of long-range correlated noise.
\newblock {\em arXiv preprint arXiv:2105.12589}, 2021.

\bibitem{evans2019scalablebayesianhamiltonianlearning}
Tim~J. Evans, Robin Harper, and Steven~T. Flammia.
\newblock Scalable {Bayesian} {Hamiltonian} learning.
\newblock {\em arXiv preprint arXiv:1912.07636}, 2019.

\bibitem{li2020hamiltonian}
Zhi Li, Liujun Zou, and Timothy~H Hsieh.
\newblock Hamiltonian tomography via quantum quench.
\newblock {\em Physical review letters}, 124(16):160502, 2020.

\bibitem{che2021learning}
Liangyu Che, Chao Wei, Yulei Huang, Dafa Zhao, Shunzhong Xue, Xinfang Nie, Jun Li, Dawei Lu, and Tao Xin.
\newblock Learning quantum {Hamiltonians} from single-qubit measurements.
\newblock {\em Physical Review Research}, 3(2):023246, 2021.

\bibitem{HaahKothariTang2022optimal}
Jeongwan Haah, Robin Kothari, and Ewin Tang.
\newblock Optimal learning of quantum {Hamiltonians} from high-temperature {Gibbs} states.
\newblock In {\em 2022 IEEE 63rd Annual Symposium on Foundations of Computer Science (FOCS)}, pages 135--146. IEEE, 2022.

\bibitem{yu2023robust}
Wenjun Yu, Jinzhao Sun, Zeyao Han, and Xiao Yuan.
\newblock Robust and efficient {Hamiltonian} learning.
\newblock {\em Quantum}, 7:1045, 2023.

\bibitem{hangleiter2024robustlylearninghamiltoniandynamics}
Dominik Hangleiter, Ingo Roth, Jonas Fuksa, Jens Eisert, and Pedram Roushan.
\newblock Robustly learning the {Hamiltonian} dynamics of a superconducting quantum processor.
\newblock {\em arXiv preprint arXiv:2108.08319}, 2024.

\bibitem{StilckFrança2024}
Daniel Stilck~Fran{\c{c}}a, Liubov~A. Markovich, V.~V. Dobrovitski, Albert~H. Werner, and Johannes Borregaard.
\newblock Efficient and robust estimation of many-qubit {Hamiltonians}.
\newblock {\em Nature Communications}, 15(1):311, Jan 2024.

\bibitem{ZubidaYitzhakiEtAl2021optimal}
Assaf Zubida, Elad Yitzhaki, Netanel~H Lindner, and Eyal Bairey.
\newblock Optimal short-time measurements for {Hamiltonian} learning.
\newblock {\em arXiv preprint arXiv:2108.08824}, 2021.

\bibitem{BaireyAradEtAl2019learning}
Eyal Bairey, Itai Arad, and Netanel~H Lindner.
\newblock Learning a local {Hamiltonian} from local measurements.
\newblock {\em Physical review letters}, 122(2):020504, 2019.

\bibitem{bairey2020learning}
Eyal Bairey, Chu Guo, Dario Poletti, Netanel~H Lindner, and Itai Arad.
\newblock Learning the dynamics of open quantum systems from their steady states.
\newblock {\em New Journal of Physics}, 22(3):032001, 2020.

\bibitem{GranadeFerrieWiebeCory2012robust}
Christopher~E Granade, Christopher Ferrie, Nathan Wiebe, and David~G Cory.
\newblock Robust online {Hamiltonian} learning.
\newblock {\em New Journal of Physics}, 14(10):103013, 2012.

\bibitem{gu2022practical}
Andi Gu, Lukasz Cincio, and Patrick~J Coles.
\newblock Practical {Hamiltonian} learning with unitary dynamics and {Gibbs} states.
\newblock {\em Nature Communications}, 15(1):312, 2024.

\bibitem{wilde2022learnH}
Frederik Wilde, Augustine Kshetrimayum, Ingo Roth, Dominik Hangleiter, Ryan Sweke, and Jens Eisert.
\newblock Scalably learning quantum many-body {Hamiltonians} from dynamical data.
\newblock {\em arXiv preprint arXiv:2209.14328}, 2022.

\bibitem{KrastanovZhouEtAl2019stochastic}
Stefan Krastanov, Sisi Zhou, Steven~T Flammia, and Liang Jiang.
\newblock Stochastic estimation of dynamical variables.
\newblock {\em Quantum Science and Technology}, 4(3):035003, 2019.

\bibitem{Caro_2024}
Matthias~C. Caro.
\newblock Learning quantum processes and {Hamiltonians} via the {Pauli} transfer matrix.
\newblock {\em ACM Transactions on Quantum Computing}, 5(2):1–53, June 2024.

\bibitem{MobusBluhmCaroEtAl2023dissipation}
Tim M{\"o}bus, Andreas Bluhm, Matthias~C Caro, Albert~H Werner, and Cambyse Rouz{\'e}.
\newblock Dissipation-enabled bosonic {Hamiltonian} learning via new information-propagation bounds.
\newblock {\em arXiv preprint arXiv:2307.15026}, 2023.

\bibitem{HolzapfelEtAl2015scalable}
M~Holz{\"a}pfel, T~Baumgratz, M~Cramer, and Martin~B Plenio.
\newblock Scalable reconstruction of unitary processes and {Hamiltonians}.
\newblock {\em Physical Review A}, 91(4):042129, 2015.

\bibitem{HuangTongFangSu2023learning}
Hsin-Yuan Huang, Yu~Tong, Di~Fang, and Yuan Su.
\newblock Learning many-body {Hamiltonians} with {Heisenberg}-limited scaling.
\newblock {\em Physical Review Letters}, 130(20):200403, 2023.

\bibitem{dutkiewicz2023advantage}
Alicja Dutkiewicz, Thomas~E O'Brien, and Thomas Schuster.
\newblock The advantage of quantum control in many-body {Hamiltonian} learning.
\newblock {\em arXiv preprint arXiv:2304.07172}, 2023.

\bibitem{MiraniHayden2024learning}
Arjun Mirani and Patrick Hayden.
\newblock Learning interacting fermionic {Hamiltonians} at the {Heisenberg} limit.
\newblock {\em arXiv preprint arXiv:2403.00069}, 2024.

\bibitem{NiLiYing2024quantum}
Hongkang Ni, Haoya Li, and Lexing Ying.
\newblock Quantum {Hamiltonian} learning for the {Fermi-Hubbard} model.
\newblock {\em Acta Applicandae Mathematicae}, 191(1):1--16, 2024.

\bibitem{LiTongNiGefenYing2023heisenberg}
Haoya Li, Yu~Tong, Tuvia Gefen, Hongkang Ni, and Lexing Ying.
\newblock Heisenberg-limited {Hamiltonian} learning for interacting bosons.
\newblock {\em npj Quantum Information}, 10(1):83, 2024.

\bibitem{BoixoSomma2008parameter}
Sergio Boixo and Rolando~D Somma.
\newblock Parameter estimation with mixed-state quantum computation.
\newblock {\em Physical Review A—Atomic, Molecular, and Optical Physics}, 77(5):052320, 2008.

\bibitem{bakshi2024structure}
Ainesh Bakshi, Allen Liu, Ankur Moitra, and Ewin Tang.
\newblock Structure learning of {Hamiltonians} from real-time evolution.
\newblock {\em arXiv preprint arXiv:2405.00082}, 2024.

\bibitem{WangLi2024simulation}
Ke~Wang and Xiantao Li.
\newblock Simulation-assisted learning of open quantum systems.
\newblock {\em Quantum}, 8:1407, 2024.

\bibitem{odake2023universal}
Tatsuki Odake, Hl{\'e}r Kristj{\'a}nsson, Philip Taranto, and Mio Murao.
\newblock Universal algorithm for transforming hamiltonian eigenvalues.
\newblock {\em arXiv preprint arXiv:2312.08848}, 2023.

\bibitem{higgins2007entanglement}
Brendon~L Higgins, Dominic~W Berry, Stephen~D Bartlett, Howard~M Wiseman, and Geoff~J Pryde.
\newblock Entanglement-free {Heisenberg-limited} phase estimation.
\newblock {\em Nature}, 450(7168):393--396, 2007.

\bibitem{ma2024learningkbodyhamiltonianscompressed}
Muzhou Ma, Steven~T. Flammia, John Preskill, and Yu~Tong.
\newblock Learning $k$-body hamiltonians via compressed sensing, 2024.

\bibitem{Zhao2025learning}
Andrew Zhao.
\newblock Learning the structure of any hamiltonian from minimal assumptions.
\newblock In {\em Proceedings of the 57th Annual ACM Symposium on Theory of Computing}, pages 1201--1211, 2025.

\bibitem{hu2025ansatzfreehamiltonianlearningheisenberglimited}
Hong-Ye Hu, Muzhou Ma, Weiyuan Gong, Qi~Ye, Yu~Tong, Steven~T. Flammia, and Susanne~F. Yelin.
\newblock Ansatz-free hamiltonian learning with heisenberg-limited scaling, 2025.

\bibitem{Arunachalam2025testing}
Srinivasan Arunachalam, Arkopal Dutt, and Francisco Escudero~Guti{\'e}rrez.
\newblock Testing and learning structured quantum hamiltonians.
\newblock In {\em Proceedings of the 57th Annual ACM Symposium on Theory of Computing}, pages 1263--1270, 2025.

\bibitem{Odake2024higher}
Tatsuki Odake, Hl{\'e}r Kristj{\'a}nsson, Akihito Soeda, and Mio Murao.
\newblock Higher-order quantum transformations of hamiltonian dynamics.
\newblock {\em Physical Review Research}, 6(1):L012063, 2024.

\bibitem{GoldreichGoldwasserRon1998property}
Oded Goldreich, Shari Goldwasser, and Dana Ron.
\newblock Property testing and its connection to learning and approximation.
\newblock {\em Journal of the ACM (JACM)}, 45(4):653--750, 1998.

\bibitem{bluhm2024hamiltonian}
Andreas Bluhm, Matthias~C Caro, and Aadil Oufkir.
\newblock Hamiltonian property testing.
\newblock {\em arXiv preprint arXiv:2403.02968}, 2024.

\bibitem{KallaugherLiang2025hamiltonian}
John Kallaugher and Daniel Liang.
\newblock Hamiltonian locality testing via {Trotterized} postselection.
\newblock {\em arXiv preprint arXiv:2505.06478}, 2025.

\bibitem{ArunachalamDuttEscudero2024arXiv}
Srinivasan Arunachalam, Arkopal Dutt, and Francisco Escudero.
\newblock Testing and learning structured quantum hamiltonians.
\newblock {\em arXiv preprint arXiv:2411.00082}, 2024.

\bibitem{PRXQuantum.3.040305}
Yulong Dong, Lin Lin, and Yu~Tong.
\newblock Ground-state preparation and energy estimation on early fault-tolerant quantum computers via quantum eigenvalue transformation of unitary matrices.
\newblock {\em PRX Quantum}, 3:040305, Oct 2022.

\bibitem{PhysRevA.92.062315}
Shelby Kimmel, Guang~Hao Low, and Theodore~J. Yoder.
\newblock Robust calibration of a universal single-qubit gate set via robust phase estimation.
\newblock {\em Phys. Rev. A}, 92:062315, Dec 2015.

\bibitem{PhysRevX.11.011020}
Andrew~M. Childs, Yuan Su, Minh~C. Tran, Nathan Wiebe, and Shuchen Zhu.
\newblock Theory of trotter error with commutator scaling.
\newblock {\em Phys. Rev. X}, 11:011020, Feb 2021.

\bibitem{4973881}
stochasticboy321 (https://math.stackexchange.com/users/269063/stochasticboy321).
\newblock Lower bound for kullback–leibler divergence between bernoulli distributed random variables.
\newblock Mathematics Stack Exchange.
\newblock URL:https://math.stackexchange.com/q/4973881 (version: 2024-09-20).

\end{thebibliography}

\appendix

\section{Kullback-Leibler Divergence Lower Bound}

We now prove the inequality used to lower bound the KL divergence in the derivation of the intolerant emptiness testing protocol. This bound and proof were sourced from \cite{4973881}.

\begin{lemma}[Kullback-Leibler Divergence Lower Bound \cite{4973881}]\label{lem:kl}
    Given two Bernoulli distributions with probabilities $x, y$, the KL divergence between them is bounded as follows:
    $$D(x||y) \ge \begin{cases}
        \frac{(x - y)^2}{2y} & x \le y \\
        \frac{(x - y)^2}{2x} & x \ge y
    \end{cases}$$
\end{lemma}

\begin{proof}
    We define the KL divergence for two Bernoulli distributions with parameters $x, y$ as follows:
    $$D(x||y) = x \ln \frac{x}{y} + (1 - x)\ln\frac{1 - x}{1 - y} = f(x, y)$$
    We start by discussing the edge cases. First, we consider the case where both $x, y$ are either 0 or 1. We have that $D(0||0) = D(1||1) = 0$ is trivially bounded by the lower bound we derived. The same is true for $D(1||0) = D(0||1) = +\infty$. If $y \in \{0, 1\}$, but $x \in (0, 1)$, we have that $D(x||y) = +\infty$ while the lower bound given is finite, which is also satisfactory. Lastly, we consider the case where $x \in \{0, 1\}, y \in (0, 1)$, where we get the following bounds:
    \begin{align*}
        D(0||y) &= \ln\frac{1}{1-y} \ge y \ge \frac{y}{2} \\
        D(1||y) &= \ln\frac{1}{y} \ge 1 - y \ge \frac{(1 - y)^2}{2},
    \end{align*}
    where the RHS represents the lower bound in the lemma.
    
    Now that we have addressed the edge cases, we now assume $x, y \in (0, 1)$. Computing the partial derivatives w.r.t. $y$, we have the following:
    \begin{align*}
        f_x(x, y) &= \ln \frac xy - \ln\frac{1 - x}{1 - y} & f_y(x, y) &= -\frac{x}{y} + \frac{1 - x}{1 - y} \\
        f_{xx}(x, y) &= \frac{1}{x} + \frac{1}{1 - x} & f_{yy}(x, y) &= \frac{x}{y^2} + \frac{1 - x}{(1 - y)^2}
    \end{align*}

    We first consider $x \le y$. Fixing $y$, the Lagrange form of Taylor's Theorem states that
    $$f(x, y) = f(y, y) + f_y(y, y)(x - y) + \frac{f_{xx}(\xi, y)}{2}(x - y)^2 = \frac{f_{xx}(\xi, y)}{2}(x - y)^2$$
    for $\xi \in [x, y]$. We can then bound the second derivative as follows:
    $$f_{xx}(\xi, y) \ge \frac{1}{\xi} \ge \frac{1}{y},$$
    meaning that
    $$f(x, y) = D(x||y) \ge \frac{(x - y)^2}{2y}, \quad x \le y$$

    To obtain the other bound, we instead assume $x \ge y$ and fix $x$, giving us the following:
    $$f(x, y) = f(x, x) + f_y(x, x)(y - x) + \frac{f_{yy}(x, \eta)}{2}(y - x)^2 = \frac{f_{yy}(x, \eta)}{2}(y - x)^2$$
    where $\eta \in [y, x]$. We can bound the second partial derivative as follows:
    $$f_{yy}(x, \eta) \ge \frac{x}{\eta^2} \ge \frac{1}{x},$$
    meaning that
    $$f(x, y) = D(x||y) \ge \frac{(x - y)^2}{2x}, \quad x \ge y$$
\end{proof}

\end{document}